\documentclass{article}
\usepackage{graphicx} 

\usepackage{amsmath}
\usepackage{amssymb}
\usepackage{amsthm}
\usepackage{blindtext}
\usepackage{titlesec}
\usepackage{ebproof}
\usepackage{expl3}
\usepackage{xparse}
\usepackage{stmaryrd}
\usepackage{algpseudocode}
\usepackage{galois}
\usepackage{float}
\usepackage{tikz}
\usetikzlibrary{decorations.pathmorphing}
\usetikzlibrary{decorations.markings}
\usetikzlibrary{calc}
\usepackage{mathtools}
\usetikzlibrary{shapes.geometric, arrows}
\usepackage[framemethod=tikz,xcolor=true]{mdframed}
\newmdenv[%
    leftmargin=0.5cm,
    tikzsetting={line width=2.0pt}%
    ]{SpecialText}%
\tikzset{->-/.style={decoration={
  markings,
  mark=at position #1 with {\arrow{>}}},postaction={decorate},thick}}

\newenvironment{myitemize}
{ \begin{itemize}
    \setlength{\itemsep}{0pt}
    \setlength{\parskip}{0pt}
    \setlength{\parsep}{0pt}     }
{ \end{itemize}                  } 

\theoremstyle{theorem}
\newtheorem{proposition}{Proposition}[section]
\theoremstyle{theorem}
\newtheorem{lemma}{Lemma}[section]
\theoremstyle{theorem}
\newtheorem{theorem}{Theorem}[section]
\theoremstyle{theorem}
\newtheorem{corollary}{Corollary}[section]
\theoremstyle{definition}
\newtheorem{definition}{Definition}[section]

\newcommand{\ifthen}[3]{\textbf{if } #1 \textbf{ then } #2 \textbf{ else } #3}
\renewcommand{\whiledo}[2]{\textbf{while } #1 \textbf{ do } #2}
\newcommand*{\defeq}{\stackrel{\text{def}}{=}}
\newcommand{\dom}[1]{\operatorname{dom}(#1)}
\newcommand{\upeq}{\stackrel{\uparrow}{=}}
\NewDocumentCommand{\up}{som}{%
  \IfBooleanTF{#1}
    {\upext{#3}}
    {\uparrow#3\IfNoValueTF{#2}{\mathord}{#2}}%
}
\NewDocumentCommand{\upext}{m}{%
  \uparrow\mleft.\kern-\nulldelimiterspace#1\mright
}

\usepackage{caption}
\usepackage{hyperref}

\hypersetup{
    colorlinks,
    linkcolor={blue},
    citecolor={blue},
    urlcolor={blue}
}

\usepackage{titling}

\title{Hal: A Language-General Framework for Analysis of User-Specified Monotone Frameworks}
\author{Abdullah Rasheed\\\small Kansas State University}
\date{ }

\begin{document}

\maketitle

\tableofcontents

\section{Introduction}

In Static Program Analysis, we are interested in finding properties of a given program. This is especially important in compiler-level optimizations, program verification, and much more. An issue that is faced, however, is language and domain generality.

For a given program property, we must write a new analyzer for each independent language because of their unique syntax and semantics. Furthermore, we must write a new program analyzer (in each language) for each new program property.

The language issue can be solved by allowing users to specify their own language syntax, and generating a parser and program analyzer given the syntax. To solve the domain issue (in a language-general manner), a framework for approaching dataflow analyses (a very large class of program analyses, later explained) is proposed in this paper. This framework aims to make user-specification of dataflow analyses in any language simple, and minimize the amount of specification required by the user (by inferring as much as possible, avoiding complications for the user).

\textbf{Section 2} covers the necessary mathematical background. While \textbf{Section 2.2-2.4} may be easier to understand with a background in grammars, such a background is unnecessary.

\textbf{Section 3} covers a background in \emph{intraprocedural} dataflow analysis, following closely (with minor differences) to \cite{popa}. This section is not entirely necessary, however it provides a very helpful intuition for the remainder of the paper. All that is needed from this section is the definitions at the beginning.

\textbf{Section 4} covers the proposed approach to \emph{interprocedural} dataflow analyses, where we vary from \cite{popa}. \textbf{Section 4.2.1} is closely accompanied by \textbf{Section 5.1}. \textbf{Section 4.2.3} covers the essential inferences and many important assumptions.

\textbf{Section 5} covers the approach to the ``new'' property space that is actually used in the implementation. It is also discussed why this approach is used as opposed to that in \textbf{Section 4.2.1}. Finally, the section ends by demonstrating that the two approaches actually work quite well with each other.

\section{Preliminaries}

In this section we present many basic definitions and results that are used throughout the paper. We will not prove any results in this section, but encourage the readers to prove the results themselves.

\newpage
\subsection{The Basics}

\begin{definition}
    A \emph{partial order} $\sqsubseteq$ over $S$ is a relation on $S$ satisfying
    \begin{myitemize}
        \item Reflexivity ($s \sqsubseteq s$)
        \item Antisymmetry ($(s \sqsubseteq s' \wedge s' \sqsubseteq s) \implies s = s'$)
        \item Transitivity ($(s_1\sqsubseteq s_2 \wedge s_2\sqsubseteq s_3) \implies s_1\sqsubseteq s_3$)
    \end{myitemize}
\end{definition}

\begin{definition}
    A \emph{partially-ordered} set (or \emph{poset}) $(S,\sqsubseteq_S)$ is a set $S$ equipped with a partial order $\sqsubseteq_S$ over $S$. We include the subscript on $\sqsubseteq$ to avoid ambiguities when working with multiple posets, however when there is no fear of ambiguity we will feel free to leave off the subscript.
\end{definition}

\begin{definition}
    Let $(S,\sqsubseteq)$ be a poset. An element $s \in S$ is an \emph{upper bound} of $P \subseteq S$ if $p \sqsubseteq s$ for all $p \in P$. Furthermore, $s$ is a \emph{least upper bound} of $P$ if $s \sqsubseteq s'$ for any upper bound $s'$ of $P$. A \emph{lower bound} and \emph{greatest lower bound} is defined dually.
\end{definition}

\begin{lemma}
    If a subset of a poset has a least upper bound, then it (the least upper bound) is unique. This is also true for greatest lower bounds.
\end{lemma}

The unique least upper bound of a subset $P \subseteq S$ is denoted $\bigsqcup_S P$ (called the \emph{join}) when it exists. This can also be written infix, where $s \sqcup_S s' = \bigsqcup_S \{s,s'\}$. We again feel free to leave off the subscript from $\bigsqcup$ and $\sqcup$ when there is no fear of ambiguity. The operators $\bigsqcap,\sqcap$ are defined dually for the unique greatest lower bound.

\begin{definition}
    A poset $(S,\sqsubseteq)$ is a \emph{complete lattice} if every subset $P \subseteq S$ has a least upper bound and a greatest lower bound. Then, $\bigsqcup\emptyset = \bigsqcap S = \bot_S$ is the \emph{least element} (also called the \emph{bottom}) and $\bigsqcap \emptyset = \bigsqcup S = \top_S$ is the \emph{greatest element} (also called the \emph{top}).
\end{definition}

Throughout this paper, we will refer to complete lattices simply as lattices.

\begin{definition}
    A subset $P \subseteq L$ of a lattice $(L, \sqsubseteq)$ is a \emph{chain} if each of its elements are comparable. That is, $\forall l_1,l_2 \in P$, $l_1 \sqsubseteq l_2$ or $l_2 \sqsubseteq l_1$. The elements of $P=(l_i)_{i \in \mathcal{I}}$ form a ``chain'' in the sense that $\cdots \sqsubseteq l_{i} \sqsubseteq l_{i'} \sqsubseteq \cdots$ for some arrangement $\ldots,i,i',\ldots$ of $\mathcal{I}$.
\end{definition}

\begin{definition}
    A lattice $L$ satisfies the \emph{Ascending Chain Condition} if each chain $l_1 \sqsubseteq l_2 \sqsubseteq \cdots$ in $L$ eventually stabilizes. That is, $\exists n$ such that $l_n = l_{n+1} = \cdots$.
\end{definition}

\begin{definition}
    Let $(L,\sqsubseteq)$ be a poset, and $P\subseteq L$ where $m = \bigsqcup P$ exists. If $m \in P$, then we say $\operatorname{max}_{\sqsubseteq}(P)$ exists, and $\operatorname{max}_{\sqsubseteq}(P) = m$.
\end{definition}

\begin{definition}
    The set of all functions $f : A \rightarrow B$ is denoted simply by $A \rightarrow B$. The set of partial functions $f : A \rightharpoonup B$ from $A$ to $B$ is denoted by $A \rightharpoonup B$. Furthermore, $(A \rightarrow B) \subseteq (A \rightharpoonup B)$.
\end{definition}

\begin{definition}
    For $f \in A \rightharpoonup B$, $\dom{f} = \{a \mid f(a) \text{ is defined}\}$. Also, $\operatorname{Im}(f) = \{f(a) \mid a \in \dom{f}\}$.
\end{definition}

\begin{definition}
    A function $f : L_1 \rightarrow L_2$ over posets $(L_1,\sqsubseteq_1),(L_2,\sqsubseteq_2)$ is called \emph{monotone} or \emph{order-preserving} if for all $l,l' \in L_1$, $l \sqsubseteq_1 l' \implies f(l) \sqsubseteq_2 f(l')$.
\end{definition}

\subsection{SimpleHal Syntax}
Here we introduce the syntax of a basic programming language that will be referenced throughout this section in examples.

\subsubsection{Specifications}

We write syntax domains beginning with an uppercase letter in bold for clarity.

\begin{figure}[H]
    \centering
    \begin{SpecialText}
    Syntax Domains
    \begin{align*}
        n &\in \textbf{Num}
        & a &\in \textbf{AExp}\\
        x &\in \textbf{Var}
        & b &\in \textbf{BExp}\\
        c &\in \textbf{Comm}
    \end{align*}
    Syntax Grammars
    \begin{align*}
        \textbf{Num} &\rightarrow 0 \mid 1 \mid \cdots\\
        \textbf{AExp} &\rightarrow n \mid x \mid -a \mid a_1 + a_2 \mid a_1 - a_2 \mid a_1 \times a_2 \mid a_1 / a_2\\
        \textbf{BExp} &\rightarrow \neg b \mid b_1 \wedge b_2 \mid b_1 \vee b_2 \mid a_1 = a_2 \mid a_1 > a_2 \mid a_1 \geq a_2\\
        \textbf{Comm} &\rightarrow x := a \mid \operatorname{read}(x) \mid c_1;c_2 \mid \textbf{if } b \textbf{ then } c_1 \textbf{ else } c_2 \\ &\hspace{12pt}\mid \textbf{while } b \textbf{ do } c
    \end{align*}
\end{SpecialText}
    \caption{SimpleHal Grammar}
    \label{fig:simplehalgrmr}
\end{figure}

It is important to note that the grammar in Fig. \ref{fig:simplehalgrmr} is certainly ambiguous, and would not be feasible for use in a parser. Parsing is not a focus in this paper, so this issue does not affect us. We may assume typical associativity (and note that ``;'' in $\mathbf{Comm}$ may associate either to the left or the right, both being equivalent; that is, $c_1;(c_2;c_3) \equiv (c_1;c_2);c_3$).

We will feel free to use $\operatorname{\mathbf{Comm}}(P)$ to denote all commands in a given program $P$ (and similarly for all other syntax domains).

\subsection{Labeling Programs}

For our analyses, we will find it convenient to label points in a program where control may eventually be passed to. Throughout this paper, we will call these points \emph{blocks}. The set of all blocks appearing in a labeled program $P$ is denoted $\operatorname{blocks}(P)$.

To illustrate these notions, consider the program
\begin{equation*}
    x\coloneqq 3;\;\operatorname{read(y)};\;\ifthen{x > y}{x\coloneqq x - 1}{y\coloneqq y-1}
\end{equation*}
We may label this program in the following way:
\begin{equation*}
    [x\coloneqq 3]_1;\; [\operatorname{read}(y)]_2;\; \ifthen{[x>y]_3}{[x\coloneqq x-1]_4}{[y\coloneqq y-1]_5}
\end{equation*}
Observe that (assuming an intuitive semantics for the language) only the components of the program are labeled. These components are the blocks of the program, as ``control'' may fall upon any of them. 

This may seem too informal, which is entirely correct. We leave the definition of blocks vague because each language may differ in what ought to be considered a block (and this may even differ from analysis to analysis in the same language). That being said, the user will be free to consider whatever they may choose as blocks in their language of choice. The only assumption we make regarding blocks here is that there are only a finite number of blocks in a given program, which is a reasonable assumption to make.

With blocks ``externally'' defined, we may provide some formalization to the notion of labelings.
\begin{definition}
    A \emph{$S$-labeling} on a program $P$ is a function $f : S \rightarrow \operatorname{blocks}(P)$
\end{definition}
Throughout this program, we will assume each program to be labeled (semantically defined by the user in practice). It is, however, cumbersome to keep defining a function as the labeling. To avoid this, we will use the notation $\llbracket B\rrbracket_{\ell}$ to mean \emph{the} block $B \in \operatorname{blocks}(P)$ that is labeled by $\ell \in S$ (that is, the unique $B \in \operatorname{blocks}(P)$ such that $f(\ell) = B$ where $f$ is the labeling on $P$).

To remain faithful to the implementation, one could instead define \emph{partial} labelings, however the analyzer knows the domain for which the function is defined on. That is, for a labeling $f$, the analyzer is only ``aware of'' $\dom{f}$, and thus it is treated essentially as a total function in $\dom{f} \rightarrow \operatorname{blocks}(P)$.

\subsection{Program Flows}

Now we would like to see how control \emph{flows} between blocks in a program. In particular, for a $S$-labeled program $P$, we would like to define a set $\operatorname{flow}(P) \subseteq S \times S$ where $(\ell,\ell') \in \operatorname{flow}(P)$ if and only if control may pass directly from $\llbracket B\rrbracket_{\ell}$ to $\llbracket B'\rrbracket_{\ell'}$. Visually, $\operatorname{flow}(P)$ is just a flow chart for $P$. For the remainder of this section, we will illustrate this idea using SimpleHal.

First, allow us to formally define the start and end points of a given program $P \in \mathbf{Comm}$.
\begin{figure}[H]
    \centering
    \begin{SpecialText}
    Initial program labels definition
    \begin{align*}
        \operatorname{init}([x := a]_{\ell}) &= \ell\\
        \operatorname{init}([\operatorname{read}(x)]_{\ell}) &= \ell\\
        \operatorname{init}(c_1;c_2) &= \operatorname{init}(c_1)\\
        \operatorname{init}(\ifthen{[b]_{\ell}}{c_1}{c_2}) &= \ell\\
        \operatorname{init}(\whiledo{[b]_{\ell}}{c}) &= \ell
    \end{align*}
    Final program labels definition
    \begin{align*}
        \operatorname{final}([x := a]_{\ell}) &= \{\ell\}\\
        \operatorname{final}([\operatorname{read}(x)]_{\ell}) &= \{\ell\}\\
        \operatorname{final}(c_1;c_2) &= \operatorname{final(c_2)}\\
        \operatorname{final}(\ifthen{[b]_{\ell}}{c_1}{c_2}) &= \operatorname{final}(c_1) \cup \operatorname{final}(c_2)\\
        \operatorname{final}(\whiledo{[b]_{\ell}}{c}) &= \{\ell\}
    \end{align*}
\end{SpecialText}
    \caption{Defining initial and final program labels in SimpleHal}
    \label{fig:initlabels}
\end{figure}

As previously stated, we are interested in how the program flows, i.e. how data/control flows through the program. From a labeled program $P$, we construct a ``flow chart'' $\operatorname{flow}(P)$ with the definitions in Fig. \ref{fig:flowdef}.
\begin{figure}[H]
    \centering
    \begin{SpecialText}
    \begin{tabular}{l l l}
        $\operatorname{flow}([x := a]_{\ell})$ & $=$ & $\emptyset$ \\[4pt]
        $\operatorname{flow}([\operatorname{read}(x)]_{\ell})$ & $=$ & $\emptyset$ \\[4pt]
        $\operatorname{flow}(c_1;c_2)$ & $=$ & $\operatorname{flow}(c_1) \cup \operatorname{flow}(c_2)$\\&&$\cup \{(\ell, \operatorname{init}(c_2)) \mid \ell \in \operatorname{final}(c_1)\}$ \\[4pt]
        $\operatorname{flow}(\ifthen{[b]_{\ell}}{c_1}{c_2})$ & $=$ & $\operatorname{flow}(c_1) \cup \operatorname{flow}(c_2)$\\&&$\cup \{(\ell, \operatorname{init}(c_1)), (\ell, \operatorname{init}(c_2))\}$\\[4pt]
        $\operatorname{flow}(\whiledo{[b]_{\ell}}{c})$ & $=$ & $\operatorname{flow}(c) \cup \{(\ell, \operatorname{init}(c))\}$\\&&$\cup \{(\ell',\ell) \mid \ell' \in \operatorname{final}(c)\}$
    \end{tabular}
\end{SpecialText}
    \caption{Defining the flow of a SimpleHal program}
    \label{fig:flowdef}
\end{figure}

Finally, we define
\begin{equation*}
    \operatorname{flow}^{-1}(P) \defeq \{(\ell', \ell) \mid (\ell, \ell') \in \operatorname{flow}(P)\}
\end{equation*}
That is, $\operatorname{flow}^{-1}(P)$ is the \emph{backwards} flow of $P$.

Let us do one concrete example. Consider the program $P$ from earlier:
\begin{equation*}
    [x\coloneqq 3]_1;\; [\operatorname{read}(y)]_2;\; \ifthen{[x>y]_3}{[x\coloneqq x-1]_4}{[y\coloneqq y-1]_5}
\end{equation*}
Observe that $\operatorname{init}(P) = 1$ and $\operatorname{final}(P) = \{3,4\}$. Now, we may follow our definition for $\operatorname{flow}$ in SimpleHal to find that
\begin{equation*}
    \operatorname{flow}(P) = \{(1,2),(2,3),(3,4),(3,5)\}
\end{equation*}

\section{Intraprocedural Analysis}

This section will move by quickly to save time for the section on interprocedural analyses, which are the main focus. This section may serve as a useful contextualization, however if you are in a rush, feel free to simply read the definition of a Monotone Framework, look at the diagrams, and move on to interprocedural analyses. 

Here we will define an abstract framework for an intraprocedural dataflow analysis. A \emph{Monotone Framework} consists of
\begin{myitemize}
    \item A lattice $(L, \sqsubseteq)$ of dataflow properties/values
    \item A set of start points $S$
    \item A set of program flows $F$
    \item An initial dataflow value $\lambda \in L$
    \item A transfer function $f_{\ell} : L \rightarrow L$ for each program point $\ell$
\end{myitemize}

These together form a Monotone Framework $(L,S,F,\lambda,\mathcal{F})$ where $\mathcal{F}$ consists of all transfer functions. While \cite{popa} imposes identity and closure of composition requirements on $\mathcal{F}$, we relax these requirements and only impose the monotonicity requirement (all $f_{\ell} \in \mathcal{F}$ is monotone) to be used much later.

We additionally impose that $L$ satisfies the Ascending Chain Condition for purposes of decidability. This will later be elaborated.

The idea is that $S$ is either the start point of the program or the set of end points, and we work forward (using the forward flows as $F$) or backward (using the backward flows as $F$) respectively. The initial dataflow value is the property that a program should start or end with (depending on what $S$ was chosen). Finally, as we move through the flows of $F$ (starting from $S$), we use the transfer functions to update the dataflow information as necessary at the given program point.

A crucial intuition for $L$ is that its partial ordering $\sqsubseteq$ is an ordering on the amount of information/precision each element in $L$ carries. That is, if two dataflow values/properties $l_1,l_2\in L$ are related in the information they carry, but $l_2$ is more abstract than $l_1$ (that is to say that $l_1$ is more precise or concrete), then we have that $l_1 \sqsubseteq l_2$ (and vice versa).

To provide a little context, we will instead consider two general frameworks: forward intraprocedural dataflow analyses and backward ones.

\subsection{Forward Dataflow Analysis}

Here, we still require the same components, but we may contextualize them a little better. Let $P$ be our program. Then, we take the previously defined framework but with $S = \operatorname{init}(P)$ and $F = \operatorname{flow}(P)$.

Now, we may define the equations for solving such an analysis:
\begin{align*}
    \operatorname{DF}_{entry}(\ell) &= \begin{cases}
        \lambda \sqcup \bigsqcup_{(\ell',\ell) \in \operatorname{flow}(P)} \operatorname{DF}_{exit}(\ell') & \ell \in \operatorname{init}(P)\\[5pt]
        \bigsqcup_{(\ell',\ell) \in \operatorname{flow}(P)} \operatorname{DF}_{exit}(\ell') & \ell \notin \operatorname{init}(P)
    \end{cases}\\
    \operatorname{DF}_{exit}(\ell) &= f_{\ell}(\operatorname{DF}_{entry}(\ell))
\end{align*}
While this may seem awkward to digest, let's take a look at how this may look visually. Consider the program
\begin{equation*}
    [x := 3]_1;[y := 4]_2;\whiledo{[x > 1]_3}{([x := x - 1]_4;[y := x \times y]_5)}
\end{equation*}
with the flow shown in Fig. \ref{fig:flowchartdesc}.

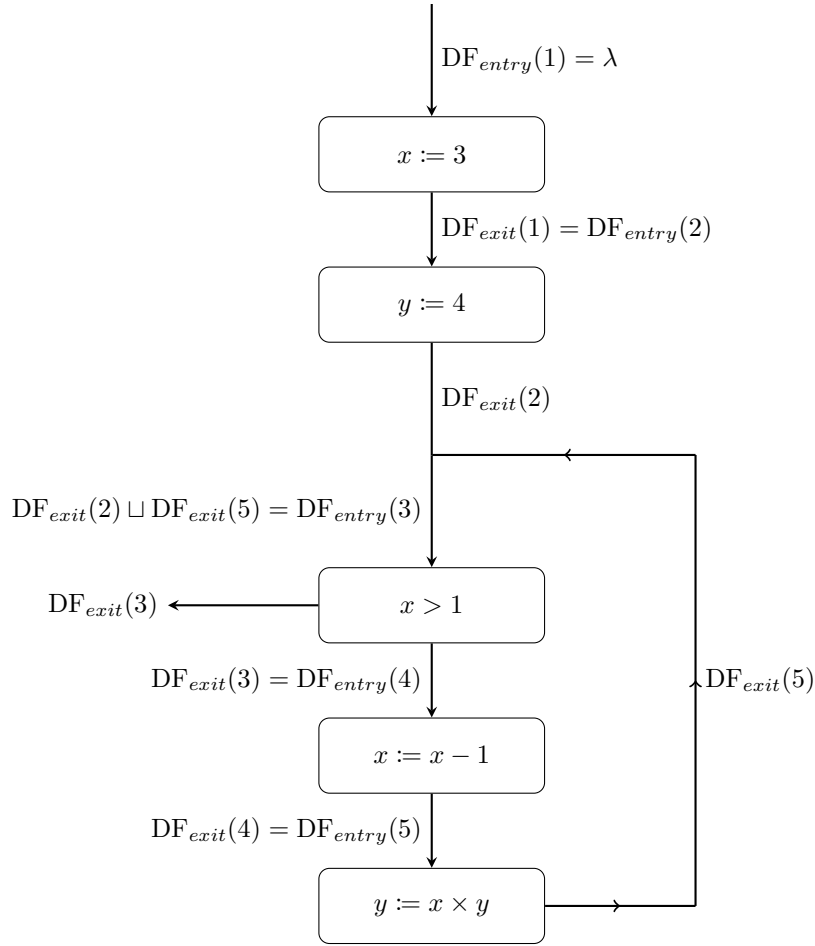
\begin{figure}[h]
    \centering
    \tikzstyle{arrow} = [thick,->,>=stealth]
\tikzstyle{line} = [thick,-,>=stealth]
\tikzstyle{block} = [rectangle, rounded corners, minimum width=3cm, minimum height=1cm,text centered, draw=black]
\begin{tikzpicture}[node distance=2cm]
\node (1) [block] {$x \coloneqq 3$};
\node (2) [block, below of=1] {$y \coloneqq 4$};
\node (3) [block, below of=2, yshift=-2cm] {$x > 1$};
\node (4) [block, below of=3] {$x \coloneqq x - 1$};
\node (5) [block, below of=4] {$y \coloneqq x \times y$};

\draw [arrow] ++(0,2) -- node[anchor=west] {$\operatorname{DF}_{entry}(1) = \lambda$} (1.north);
\draw [arrow] (1) -- node[anchor=west] {$\operatorname{DF}_{exit}(1) = \operatorname{DF}_{entry}(2)$} (2);
\draw [line] (2.south) -- node[anchor=west] {$\operatorname{DF}_{exit}(2)$} ++(0,-1.5);
\draw [->-=.5] (5.east) -- ++(2,0);
\draw [->-=.5] (5.east)+(2,0) -- node[anchor=west] {$\operatorname{DF}_{exit}(5)$} ++(2,6);
\draw [->-=.5] (5.east)+(2,6) -- ++(-1.5,6);
\draw [arrow] (2.south)+(0,-1.5) -- node[anchor=east] {$\operatorname{DF}_{exit}(2) \sqcup \operatorname{DF}_{exit}(5) = \operatorname{DF}_{entry}(3)$} ++(3);
\draw [arrow] (3) -- node[anchor=east] {$\operatorname{DF}_{exit}(3) = \operatorname{DF}_{entry}(4)$} (4);
\draw [arrow] (4) -- node[anchor=east] {$\operatorname{DF}_{exit}(4) = \operatorname{DF}_{entry}(5)$} (5);
\draw [arrow] (3.west) -- node[anchor=east,xshift=-1cm] {$\operatorname{DF}_{exit}(3)$} ++(-2,0);
\end{tikzpicture}
    \vspace{7pt}\caption{Flow chart of a program}
    \label{fig:flowchartdesc}
\end{figure}

Observe how each block takes the preceding data flow information, changes it in a way consistent with the analysis (using the blocks transfer function $f_{\ell}$, which is not explicitly stated in the diagram), and passes it on. In the case where there are multiple points of entry for a block, we have to merge/combine the data flow information from each point to get the best and most correct possible data flow information at the entry of that block.

For concrete examples of forward dataflow analyses, refer to \cite{popa}.

\subsection{Backward Dataflow Analysis}

Similar to the previous section, we will show how Monotone Frameworks can be used to define \emph{backward} analyses. Let $P$ be our program. We now set $S = \operatorname{final}(P)$ and $F = \operatorname{flow}^{-1}(P)$. So we instead start at the end and work our way backwards. The flow chart will look the same as in the previous section except that the arrows are all reversed.

The equations here now become
\begin{align*}
    \operatorname{DF}_{exit}(\ell) &= \begin{cases}
        \lambda \sqcup \bigsqcup_{(\ell',\ell) \in \operatorname{flow}^{-1}(P)} \operatorname{DF}_{exit}(\ell') & \ell \in \operatorname{final}(P)\\[5pt]
        \bigsqcup_{(\ell',\ell) \in \operatorname{flow}^{-1}(P)} \operatorname{DF}_{exit}(\ell') & \ell \notin \operatorname{final}(P)
    \end{cases}\\
    \operatorname{DF}_{entry}(\ell) &= f_{\ell}(\operatorname{DF}_{entry}(\ell))
\end{align*}
For concrete examples of backward dataflow analyses, refer to \cite{popa}.

\subsection{Solving Intraprocedural Dataflow Analyses}

It is clear that the equations for solving intraprocedural dataflow analyses follow the same pattern. That is,
\begin{align*}
    \operatorname{DF}_{\circ}(\ell) &= \lambda_S^{\ell} \sqcup \bigsqcup_{(\ell',\ell) \in F} \operatorname{DF}_{\bullet}(\ell')\\
    \operatorname{DF}_{\bullet}(\ell) &= f_{\ell}(\operatorname{DF}_{\circ}(\ell))
\end{align*}
where
\begin{equation*}
    \lambda_S^{\ell} = \begin{cases}
        \lambda & \text{if } l \in S\\
        \bot & \text{otherwise}
    \end{cases}
\end{equation*}
A intuitive mathematical solution to this is the \emph{Meet Over all Paths} (MOP) solution. Although it is called the \emph{Meet} Over all Paths, we are instead taking the join over all paths (simply a philosophical difference; we define things dually to how they are historically defined). The idea behind this solution is to ``follow'' the transfer functions through the flows that lead up to $\ell$ for each point $\ell$.

We do not go in depth into this solution because MOP is unfortunately an undecidable problem. \cite{Kam1977} proves that if the MOP solution is decidable, then so is the so-called \emph{Modified} Post-Correspondence Problem (MPCP). \cite{hopcroft1969formal} shows MPCP is undecidable, and so MOP is undecidable.

Instead what we will resort to is the \emph{Maximal Fixed Point} (MFP) solution (which, in a similar fasion to MOP, is backwards for us; we are instead computing the \emph{minimal} fixed point), which is a decidable approximation to the correct MOP solution. An algorithmic solution to MFP is outlined in \cite{popa}, as well as a proof that it safely approximates MOP (both following those of \cite{Kam1977}).

Later in this paper we will be modifying the algorithm outlined in \cite{popa} to suit the new machinery we introduce in the following section.

\section{Interprocedural Analysis}

\subsection{The Struggles}

Suppose we have procedures/functions in our language (which we will later define). Attempting to solve a dataflow analysis now becomes a little more complicated. Let us cover a few of the reasons that might cause us to think a little harder than we would for an intraprocedural analysis.

\subsubsection{Invalid Paths}

Consider a program represented by the flow chart shown in Fig. \ref{fig:interprocflow}.

\begin{figure}[h]
    \centering
    \tikzstyle{arrow} = [thick,->,>=stealth]
\tikzstyle{line} = [thick,-,>=stealth]
\tikzstyle{block} = [rectangle, rounded corners, minimum width=3cm, minimum height=1cm,text centered, draw=black]
\tikzstyle{text} = [rectangle, rounded corners, minimum width=3cm, minimum height=1cm,text centered, draw=none]
\begin{tikzpicture}[node distance=1.5cm]
    \node (alabel) {\textbf{Procedure `A'}};
    \node (a1) [block, below of=alabel, yshift=0.5cm] {a1};
    \node (a2) [block, below of=a1] {a2 [Call \textbf{A}]};
    \node (a3) [block, below of=a2] {a3 [Return from \textbf{A}]};
    \node (a4) [block, below of=a3] {a4 [End \textbf{A}]};
    \node (b1) [block, right of=a1, xshift=4cm] {b1};
    \node (blabel) [above of=b1, yshift=-0.5cm] {\textbf{Procedure `B'}};
    \node (b2) [block, below of=b1] {b2 [if statement]};
    \node (b3) [block, below of=b2, right of=b2, xshift=1cm] {b3};
    \node (b4) [block, below of=b3] {b4 [Call \textbf{B}]};
    \node (b5) [block, below of=b4] {b5 [Return from \textbf{B}]};
    \node (b6) [block, below of=b5, left of=b5, xshift=-1cm] {b6 [End \textbf{B}]};

    \draw [arrow] (a1) -- (a2);
    \draw [arrow] (a2) -- (a3);
    \draw [arrow] (a3) -- (a4);
    \draw [line] (a2.east) -- ++(1,0);
    \draw [line] (a2.east)+(1,0) -- ++(1,1.5);
    \draw [arrow] (a2.east)+(1,1.5) -- (b1);
    \draw [arrow] (b1) -- (b2);
    \draw [arrow] (b2) -- (b3);
    \draw [arrow] (b2) -- (b6);
    \draw [arrow] (b3) -- (b4);
    \draw [line] (b4.east) -- ++(1,0);
    \draw [line] (b4.east)+(1,0) -- ++(1,4.5);
    \draw [arrow] (b4.east)+(1,4.5) -- (b1);
    \draw [arrow] (b4) -- (b5);
    \draw [arrow] (b5) -- (b6);
    \draw [line] (b6.west) -- ++(-1.5,0);
    \draw [line] (b6.west)+(-1.5,0) -- ++(-1.5,4.5);
    \draw [arrow] (b6.west)+(-1.5,4.5) -- (a3);
    \draw [line] let \p1 = ([xshift=1cm]b5.east), \p2 = (b6.east) in (\p2) -- (\x1,\y2) (\x1,\y2) -- (\p1) [arrow] (\p1) -- (b5.east);
\end{tikzpicture}
    \vspace{7pt}\caption{Example interprocedural program flow chart}
    \label{fig:interprocflow}
\end{figure}
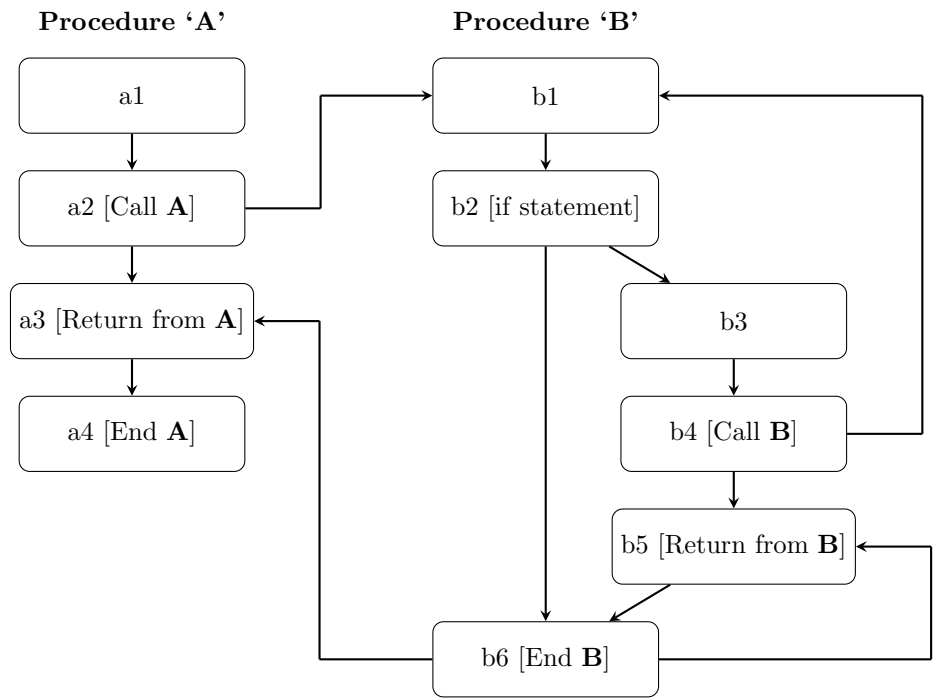

Observe that the execution trace $\text{a1} \rightarrow \text{a2} \rightarrow \text{b1} \rightarrow \text{b2} \rightarrow \text{b6} \rightarrow \text{b5}$ is not possible. The flow of the program, however, does not indicate this (as there is, technically, an arrow from b6 to b5 since $\text{b6} \rightarrow \text{b5}$ can occur in some conditions).

Therefore, attempting to solve an \emph{inter}procedural analysis the same way we solve an intraprocedural analysis will cause issues (i.e. we will likely lose valid dataflow information due to joining information from invalid paths/flows). To solve this issue, we would naturally like to find only the \emph{valid} paths (the execution traces that are actually possible).

A straightforward approach is to find only the valid paths, and compute the \emph{MVP} (\emph{Meet over Valid Paths}). Unfortunately we see that, using the arguments that MOP is undecidable, MVP is undecidable. Thus, we look back to the MFP solution, and how we may be able to implement it for interprocedural analyses while only passing data through valid paths.

\subsubsection{Procedure Dataflows Are Messy}

Suppose two different points in the program call the same function. Then, we have some dataflow value $l_1$ and some other value $l_2$ entering the same function. These two values are from completely different contexts, so it is apparent that we should not join them together (though we could, but it will yield very muddled and not very useful information). To solve this, we will reform our dataflow property space to take context into account.

\subsubsection{Analyses Are Demanding}

Transferring information to the start of a method or returning information from a method back to its calling point may call for a special transfer method. Therefore, we ought to differentiate between call, return, and intraprocedural flows and transfer functions.

\subsection{Formalizing Interprocedural Dataflow Analyses}

We first ought to observe that context is important (as earlier motivated). We need to know the `call stack' to know where certain data flow information needs to return from inside a procedure. We also need to redefine our transfer functions to carry this information. Finally, we need our transfer functions to have the flexibility to transfer data flow information differently for different types of flows (i.e. between our normal flows, procedure-calling flows, return flows, etc.). We shall achieve these things in this section with a bit of extra machinery.

We first define a set $\Delta$ of `context values' of which our data flow information will carry. It is important to observe why this is necessary. Suppose a procedure is called at multiple points in the program. It is then evident that the data flow information throughout that procedure should be different when called at one point as opposed to another (as they will enter the procedure with differing data flow information). Therefore, we see that the data flow information is dependent on the \emph{context} from which the procedure is called. Also observe that keeping track of calling contexts allows us to solve the viable paths problem mentioned earlier.

Now, we define an \emph{Embellished Monotone Framework} $(\widehat{L}, S, F, \widehat{\lambda}, \widehat{\mathcal{F}})$ from a Monotone Framework $(L, S, F, \lambda, \mathcal{F})$ with the following new definitions:
\begin{myitemize}
    \item $\widehat{L} = \Delta \rightharpoonup L$
    \item $\widehat{\lambda} \in \widehat{L}$ is the initial data flow value
    \item $\widehat{\mathcal{F}}$ consists of the new transfer functions $\widehat{f_{\ell}} : \widehat{L} \rightarrow \widehat{L}$
\end{myitemize}

It is common to see $\widehat{L} = \Delta \rightarrow L$ in literature, however for the sake of consistency with implementation we instead use the set of partial functions $\Delta \rightharpoonup L$.

Our new initial data flow value may be intuitively defined as
\begin{equation*}
    \widehat{\lambda}(\delta) = \begin{cases}
        \lambda & \text{if } \delta = \Lambda\\
        \text{undefined} & \text{otherwise}
    \end{cases}
\end{equation*}
where $\Lambda \in \Delta$ is the context value which carries no context (empty). This makes good sense since there are (reasonably assumed to be) no `pending' procedure calls at the beginning of a program (nor at the end).

Now, let us visualize what is happening to data flow information when we call a procedure.

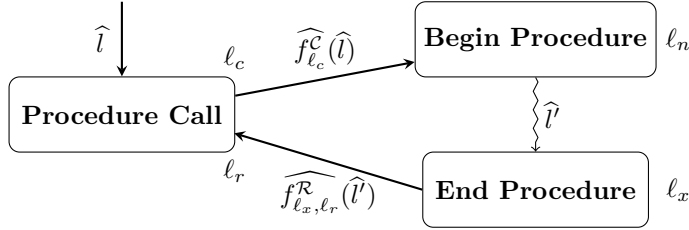
\tikzstyle{arrow} = [thick,->,>=stealth]
\tikzstyle{line} = [thick,-,>=stealth]
\tikzstyle{block} = [rectangle, rounded corners, minimum width=3cm, minimum height=1cm,text centered, draw=black]
\tikzstyle{text} = [rectangle, rounded corners, minimum width=3cm, minimum height=1cm,text centered, draw=none]
\begin{tikzpicture}[node distance=1.5cm]
    \node (proccall) [block] {\textbf{Procedure Call}};
    \node (lc) [above of=proccall, yshift=-0.75cm, xshift=1.5cm] {$\ell_c$};
    \node (lr) [above of=proccall, yshift=-2.25cm, xshift=1.5cm] {$\ell_r$};
    \node (procstart) [block, above of=proccall, right of=proccall, xshift=4cm, yshift=-0.5cm] {\textbf{Begin Procedure}};
    \node (procend) [block, below of=procstart, yshift=-0.5cm] {\textbf{End Procedure}};
    \node (ln) [right of=procstart, xshift=0.4cm] {$\ell_n$};
    \node (lx) [right of=procend, xshift=0.4cm] {$\ell_x$};

    \draw [->,
line join=round,
decorate, decoration={
    zigzag,
    segment length=7,
    amplitude=.9,post=lineto,
    post length=2pt
}]  (procstart) -- node[anchor=west] {$\widehat{l}'$} (procend);

    \draw [arrow] (proccall.east)+(0,0.25) -- node[anchor=south] {$\widehat{f_{\ell_c}^{\mathcal{C}}}(\widehat{l})$} (procstart);
    \draw [arrow] (procend.west) -- node[anchor=north] {$\widehat{f_{\ell_x,\ell_r}^{\mathcal{R}}}(\widehat{l}')$} ++(-2.5,0.75);
    \draw [arrow] ++(0,1.5) -- node[anchor=east,xshift=-0.1cm] {$\widehat{l}$} (proccall);
    
\end{tikzpicture}

For now, do not worry about these new functions, as they will be introduced soon. Essentially, we set up the function with whatever makes sense for the analysis. What is interesting in our following formalizations (and thus implementation) is that we provide a framework general enough to cover cases such as carrying over information from before the procedure to after it returns and directing dataflow information directly between the call and return points.

\subsubsection{The New Property Space: A Straightforward Approach}

Before we move on, there is still some work to be done. We still have yet to completely define $\widehat{L}$. We need to know the ordering on $\widehat{L}$, which begs the question, what does it mean for $\widehat{l} \in \widehat{L}$ to have ``more'' information than $\widehat{l}' \in \widehat{L}$? To answer this, let us take inspiration from the typical order on functions in $S \rightarrow M$ where $S$ is a non-empty set and $M$ is a (non-empty) lattice. The typical partial order is as follows:
\begin{equation*}
    f \sqsubseteq g \;\text{ iff }\; \forall s \in S,\; f(s) \sqsubseteq_M g(s)
\end{equation*}
That is, if at every point of $S$, $g$ gives at least as much information as $f$. Unfortunately, this is not very feasible for partial functions. To achieve a similar philosophy, we require containment of domains: for $f,g \in S \rightharpoonup M$
\begin{equation*}
    f \sqsubseteq_{S \rightharpoonup M} g \;\text{ iff }\; \dom{f} \subseteq \dom{g} \text{ and } \forall s \in \dom{f},\; f(s) \sqsubseteq_M g(s)
\end{equation*}
This relation is clearly a partial order, and the join operation follows as:
\begin{equation*}
    (f \sqcup g)(s) = {\uparrow}f(s) \sqcup_M {\uparrow}g(s)
\end{equation*}
defined \emph{only} on $\dom{f} \cup \dom{g}$, where
\begin{equation*}
    {\uparrow}f(s) = \begin{cases}
        f(s) & \text{if } s \in \dom{f}\\
        \bot_M & \text{otherwise}
    \end{cases}
\end{equation*}
is a total function in $S \rightarrow M$. The idea is that we join the information at each point, except for the points for which one function is undefined, in which case we assume the value of the other function at those points. For instance,
\begin{equation*}
    [x \mapsto m_1] \sqcup [x \mapsto m_2, y \mapsto m_3] = [x \mapsto m_1\sqcup_M m_2, y \mapsto m_3]
\end{equation*}
Also observe that we worked backwards. Since total functions are a subset of partial functions, it is actually the partial order of total functions that is induced from that of partial functions.
\begin{proposition}
    ${\uparrow}(f\sqcup g) = {\uparrow}f\sqcup {\uparrow}g$.
\end{proposition}
\begin{proof}
    Straightforward from definition. Left for the reader.
\end{proof}
\begin{lemma}
    $f\sqcup g$ is the least upper bound of $f,g \in S \rightharpoonup M$.
\end{lemma}
\begin{proof}
    Suppose that $h \in S \rightharpoonup M$ is an upper bound of $f,g \in S \rightharpoonup M$. Then, $\dom{f} \subseteq \dom{h}$ and $\dom{g} \subseteq \dom{h}$, so $\dom{f\sqcup g} = \dom{f}\cup\dom{g} \subseteq \dom{h}$. Also,
    \begin{align*}
        \forall s \in \dom{f},\; f(s) &\sqsubseteq_M h(s)\\
        \forall s \in \dom{g},\; g(s) &\sqsubseteq_M h(s)
    \end{align*}
    We need to prove
    \begin{equation*}
        \forall s \in \dom{f}\cup\dom{g},\; {\uparrow}f(s)\sqcup_M {\uparrow}g(s) \sqsubseteq_M h(s)
    \end{equation*}
    Suppose $s \in \dom{f} \cap \dom{g}$. Then,
    \begin{equation*}
        {\uparrow}f(s)\sqcup_M {\uparrow}g(s) = f(s)\sqcup_M g(s) \sqsubseteq_M h(s)
    \end{equation*}
    since $h(s)$ is an upper bound of $f(s)$ and $g(s)$. Now suppose $s \in \dom{f} \setminus \dom{g}$. Then,
    \begin{equation*}
        {\uparrow}f(s)\sqcup_M {\uparrow}g(s) = f(s)\sqcup_M \bot_M \sqsubseteq_M h(s)
    \end{equation*}
    The case for $s \in \dom{g}\setminus\dom{f}$ is symmetric.
\end{proof}

\begin{corollary}
    $S \rightharpoonup M$ is a (complete) lattice.
\end{corollary}
\begin{proof}
    Let $P \subseteq S \rightharpoonup M$, and define $D = \bigcup_{f \in P}\dom{f}$. Then, let
    \begin{equation*}
        p = \bigsqcup_{f\in P} f
    \end{equation*}
    We show $p$ exists in $S \rightharpoonup M$ by simple construction. First, observe that $D \subseteq S$. Then, $\forall d \in D$,
    \begin{equation*}
        p(d) = \bigsqcup_{f \in P} {\uparrow}f(s)
    \end{equation*}
    exists since $M$ is a lattice (and thus all its subsets have a least upper bound).

    Let $f \in P$ (if $P$ is empty, the following is satisfied vacuously). Then, we clearly have that $f \sqsubseteq p$ since $\dom{f} \subseteq D$ and $s \in \dom{f} \implies f(s) \sqsubseteq_M p(d)$. It is left for the reader to show that $p$ is the \emph{least} upper bound.

    It is also left for the reader to show Lemma 4.1 for $(f \sqcap g)(s) = f(s) \sqcap_M g(s)$ (defined over $\dom{f} \cap \dom{g}$) and complete this corollary.
\end{proof}
We have now proven the lattice requirement, but we still must prove one more condition to make this lattice usable (as was mentioned earlier in the paper).
\begin{lemma}
    Suppose $M$ satisfies the Ascending Chain Condition. Then, $S \rightharpoonup M$ satisfies the Ascending Chain Condition if and only if $S$ is finite.
\end{lemma}
\begin{proof}
    We begin in the backward direction: assume $S$ is finite. Let $(f_i)_{i \in \mathcal{I}}$ be a chain in $S \rightharpoonup M$ and pick an arbitrary $i,i' \in \mathcal{I}$ with $f_i \sqsubseteq f_{i+1}$. This implies $\dom{f_i} \subseteq \dom{f_{i'}}$ and
    \begin{equation*}
        \forall s \in \dom{f_i},\; f_i(s) \sqsubseteq_M f_{i'}(s)
    \end{equation*}
    Assume $f_i \neq f_{i'}$. Then there are two (not mutually exclusive) cases. The first case is that $\dom{f_i} \subsetneq \dom{f_{i'}}$. Since $S$ is finite, this can only happen finitely many times. The second case is that $\exists s \in \dom{f_i}$ such that $f_i(s)$ is strictly less than $f_{i'}(s)$. Since $M$ satisfies the Ascending Chain Condition, this can only happen finitely many times. Thus, the chain can only grow finitely many times before stabilizing, and so $S \rightharpoonup M$ satisfies the Ascending Chain Condition.
    \\\\Now we prove the forward direction by its contrapositive: assume $S$ is not finite. Then, let $(f_n)_{n \in \mathbb{N}}$ be a chain such that $\dom{f_1}$ is finite and $\dom{f_{i+1}} = \dom{f_i} \cup \{s\}$ with $s \notin \dom{f_i}$. This chain clearly does not stabilize, and thus $S \rightharpoonup M$ does not satisfy the Ascending Chain Condition. By the contrapositive, this proves the forward direction.
\end{proof}
We would now like to prove that $\widehat{L}$ satisfies the Ascending Chain Condition. In order to achieve this, however, we need to make a key assumption: call stacks are bounded. This ought to be the case in programming languages regardless to avoid stack overflows, however we impose this requirement on $\Delta$ to ensure that call contexts can not be arbitrarily long (which could result from recursion). Under this assumption, we get the following result:
\begin{corollary}
    $\widehat{L}$ satisfies the Ascending Chain Condition.
\end{corollary}
\begin{proof}
    By Lemma 4.2, since $\Delta$ is finite and $L$ satisfies the Ascending Chain Condition, we immediately have that $\widehat{L}$ satisfies the Ascending Chain Condition.
\end{proof}
While this new property space may work fine for our purposes, the Hal implementation instead uses the approach outlined in Section 5.1. In Section 5.1.1, the approaches are compared, and it is highly encouraged for the reader to take a look at these sections.

\subsubsection{Interprocedural Transfer Functions}

Now, we would like to give the user more flexibility with their analysis (with respect to the calling/returning semantics). Since we are abstracting the semantics of the language that is being evaluated (for the user's discretion), we want to cover cases where a programming language's semantics may involve wacky things while calling or returning from a function. Therefore, we will make one more slight reformation to our Embellished Monotone Framework. We now define
\begin{equation*}
    \widehat{\mathcal{F}} = (\widehat{\mathcal{F_N}}, \widehat{\mathcal{F_C}}, \widehat{\mathcal{F_R}})
\end{equation*}
where $\widehat{\mathcal{F_N}}$ is the set of \emph{normal} functions $\widehat{f_{\ell}^{\mathcal{N}}} : \widehat{L} \rightarrow \widehat{L}$, which are used for intraprocedural flows, $\widehat{\mathcal{F_C}}$ is the set of \emph{call} functions $\widehat{f_{\ell}^{\mathcal{C}}} : \widehat{L} \rightarrow \widehat{L}$, which are used for interprocedural function-calling flows, and $\widehat{\mathcal{F_R}}$ is the set of \emph{return} functions $\widehat{f_{\ell,\ell'}^{\mathcal{R}}} : \widehat{L} \rightarrow \widehat{L}$, which are used for interprocedural function-returning flows. Observe that return functions are different because it is important to know where we are returning to ($\ell'$). Otherwise, we would be moving information from contexts that have no relevance to $\ell'$. Fortunately, this ``screening'' of dataflow information can be achieved entirely through inference, which we will see in the following section.

Observe that $\widehat{f_{\ell_c}^{\mathcal{C}}}$ and $\widehat{f_{\ell_x,\ell_r}^{\mathcal{R}}}$ only need to be defined for program points in interprocedural flows (and similarly for $\widehat{\mathcal{F_N}}$ over intraprocedural flows). Nevertheless, this is typically generalizable (as will be discussed with the $\operatorname{transfer}$ function in the following section), making it much easier on the user and the developer during implementation.

Now we need a mechanism to allow us to know which functions should be used during evaluation. There are many ways to formalize this mathematically. For instance, \cite{popa} defines an interprocedural flow structure consisting of all such 4-tuples $(\ell_c,\ell_n,\ell_x,\ell_r)$ where $(\ell_c,\ell_n),(\ell_x,\ell_r) \in F$ and $\ell_n$ is the start of a function and $\ell_x$ is the end of a function, and $\ell_c,\ell_r$ are the call and return points (respectively) for the same block.

For a formalization better representing the implementation, we could instead choose to reform $F$ to a set $\widetilde{F}$ of \emph{tagged} flows. Each flow is ``tagged'' with its flow type (either normal, calling, or returning, corresponding to $\mathcal{N,C,R}$ respectively). Working with this mathematically is a little bit messy, so for the remainder of this paper we will instead choose to introduce a function which serves the same purpose. We will define $\Gamma : F \rightarrow \{\mathcal{N,C,R}\}$ as the function which ``tags'' every flow. We will also feel free to abuse notation by saying $\widehat{f_{\ell}^{\Gamma(\ell,\ell')}}$ as the function corresponding to the type of flow we are dealing with.

Observe that our formalization allows the freedom for function-calling blocks to still have normal flows. For instance, if we would like the call point of a block to flow data to its return point (for instance, if the data flow information before calling the function is of importance), that may be achieved through its normal function (while its call function may be different).

Finally, we end this portion by reiterating the monotonicity requirement on $\widehat{\mathcal{F}}$ as we did on $\mathcal{F}$.

\subsubsection{Putting It All Together}

We have finally done all that needs to be done. Sewing these pieces together leaves us with $(\widehat{L}, S, F, \Gamma, \widehat{\lambda}, \widehat{\mathcal{F}})$ as our new and improved Embellished Monotone Framework. Observe that we added $\Gamma$ for the sake of completion mathematically, however in implementation (as previously stated), each flow is attached with a tag inherently (which the user defines), so $\Gamma$ is implicitly defined/inferred from the user's specification of $\widetilde{F}$.

It also should be specified that the user will not need to individually specify transfer functions for each program point (indeed, that would defeat our whole purpose, it would then only work for a very small set of programs!). Instead, the user only needs to specify one function
\begin{equation*}
    \operatorname{transfer} : \mathbf{Block} \times \{\mathcal{N,C,R}\} \times L \rightarrow L
\end{equation*}
where $\operatorname{transfer}(\llbracket B\rrbracket_{\ell}, \Gamma(\ell,\ell'), l)$ corresponds to $\widehat{f_{\ell}^{\Gamma(\ell,\ell')}}(\delta, l)$ for some context $\delta$. We will soon address inference of contexts.

The users also have little work to do with regards to labeling. They need not specify the call and return labels. All that needs to be done is that a well-formed labeling is defined for the program. Then, for any $(\ell,\ell') \in F$ such that $\Gamma(\ell,\ell') = \mathcal{C}$, we may create a call label $\ell_c$ and return label $\ell_r$ for $\ell$ (with the call and return labels effectively replacing $\ell$). We shall feel free to stick to the subscript $c$ and $r$ to refer to a given labels call and return labels respectively given that it has call and return labels.

That note suggests the need to impose a well-formedness restriction on $\Gamma$ (or in the implementation's case, $\widetilde{F}$). That is, it must be the case that for any $(\ell,\ell') \in F$ such that $\Gamma(\ell,\ell') = \mathcal{C}$, there exists $\ell''$ such that $\Gamma(\ell'', \ell) = \mathcal{R}$. That is, every function call must eventually return back to where it was called from. With this imposition, we may safely assume that any label which calls a function or is returned to from a function will have a $\ell_c$ and $\ell_r$.

To work with call and return labels more rigorously, one would need to define functions for moving forward and backward between the user-defined labeling and the new labeling with call and return labels. We will spare ourselves of the precise details, however the idea is this: Given a program $P$ and a $\mathcal{L}$-labeling over $P$, we construct a new $\widehat{\mathcal{L}}$-labeling over $P$ where each $\ell \in \mathcal{L}$ is replaced by two new labels $\ell_c,\ell_r$ (the respective call and return labels) only if $\Gamma(\ell,\ell') = \mathcal{C}$ for some $\ell' \in \mathcal{L}$. Now, we construct a map $\rho : \widehat{\mathcal{L}} \rightarrow \mathcal{L}$ which goes ``backwards'' in the sense that $\rho(\ell_c) = \rho(\ell_r) = \ell$ (using the aforementioned $\ell,\ell_c,\ell_r$).

As for contexts, the user need not mess with contexts at all, as our machinery allows us to infer everything necessary. The inference is as follows:
\begin{align*}
    \widehat{f_{\ell}^{\mathcal{N}}}(\widehat{l})(\delta) &=  \operatorname{transfer}(\llbracket B\rrbracket_{\ell}, \mathcal{N}, \widehat{l}(\delta))\\
    \widehat{f_{\ell}^{\mathcal{C}}}(\widehat{l})(\delta) &= \begin{cases}
        \operatorname{transfer}(\llbracket B\rrbracket_{\ell}, \mathcal{C}, \widehat{l}(\delta')) & \text{if } \delta = \delta';\ell_c\\
        \text{undefined} & \text{otherwise}
    \end{cases}\\
    \widehat{f_{\ell_x,\ell_r}^{\mathcal{R}}}(\widehat{l})(\delta) &= \operatorname{transfer}(\llbracket B\rrbracket_{\ell}, \mathcal{R}, \widehat{l}(\delta;\ell_c)) \quad \text{where $\llbracket B\rrbracket_{\rho(\ell_c)} = \llbracket B\rrbracket_{\rho(\ell_r)}$}
\end{align*}
In the first case, we simply take the data flow information and transform it directly according to the user's specification without changing the call context. In the second case, we first transform the data flow information (according to the user's specification), and then push the corresponding call label to the call context. In the last case, we do the opposite: we pop the corresponding call label from the call context and then transform. The condition that $\llbracket B\rrbracket_{\rho(\ell_c)} = \llbracket B\rrbracket_{\rho(\ell_r)}$ simply means that $\ell_c$ is the call point for the block that we are returning to (point $\ell_r$).

Observe that the user-defined function ``$\operatorname{transfer}$'' must be monotone with respect to the last parameter in order to maintain the monotonicity requirement of $\widehat{\mathcal{F}}$. Fortunately, this is relatively natural as an increase in the amount of information in the input should lead to an increase in the amount of information in the output.

We now see that from $L,\lambda,\operatorname{transfer}$, we may deduce $\widehat{L},\widehat{\lambda},\widehat{\mathcal{F}}$. Thus, from an \emph{Implicit Monotone Framework} $(L, S, \widetilde{F}, \lambda, \operatorname{transfer})$, we may infer the Embellished Monotone Framework $(\widehat{L}, S, F, \Gamma, \widehat{\lambda}, \widehat{\mathcal{F}})$ required to conduct the analysis.

Another inference we may make from the analyzer which we briefly outlined earlier is program labels. The user only need define a well-formed labeling, and the analyzer can then infer all function-calling labels with $\widetilde{F}$. Then, the analyzer may split up such a label $\ell$ into its call and return labels $\ell_c,\ell_r$ and adjust $\widetilde{F}$ accordingly.

\subsection{Solving Interprocedural Analysis}

Now that we have made clear what the user ought to define and how we may infer from it an Embellished Monotone Framework, we may proceed to computing the analysis. As earlier motivated, a MOP solution is undecidable, as is a MVP solution. Therefore, we should feel safer using a MFP solution as was done for intraprocedural analyses. With the introduction of ``tagged'' flows, we will require some changes to the MFP solution presented in \cite{popa}.

\subsubsection{The Mathematical Solution}

Before we attempt to solve this algorithmically, however, we ought to understand what the mathematical solution is so we can formally prove the correctness of our algorithm. \cite{popa} provides a mathematical solution for their interprocedural analyses, however our framework differs.

For interprocedural analyses, the entry values are analogous to those in \cite{popa}:
\begin{equation*}
    \operatorname{DF}_{\circ}(\ell) = \widehat{\lambda}_E^{\ell} \;\sqcup\; \bigsqcup_{(\ell',\ell) \in F} \operatorname{DF}_{\bullet}(\ell')(\ell)
\end{equation*}
where
\begin{equation*}
    \widehat{\lambda}_S^{\ell} = \begin{cases}
        \widehat{\lambda} & \text{if } \ell \in S\\
        \bot & \text{otherwise}
    \end{cases}
\end{equation*}
Where we differ from \cite{popa} is in the exit values. Since we ``split'' flows and their respective transfer functions into different kinds, labels may have multiple outputs. \cite{popa} achieves the same ability through a generalization which serves to ``skip'' through functions by carrying dataflow values across through function calls (directly to the return point). Our exit values are then mathematically determined by:
\begin{equation*}
    \operatorname{DF}_{\bullet}(\ell)(\ell') = f_{\ell}^{\Gamma(\ell,\ell')}(\operatorname{DF}_{\circ}(\ell))
\end{equation*}
for all $\ell'$ such that $(\ell,\ell') \in F$ (and not defined otherwise).

\subsubsection{The MFP Algorithm}

In our modified algorithm, we take as input an Embellished Monotone Framework $(\widehat{L}, S, F, \Gamma, \widehat{\lambda}, \widehat{\mathcal{F}})$ (which is inferred from an Implicit Monotone Framework), and provide as output the entry/exit values $\mathrm{MFP}_{\circ},\mathrm{MFP}_{\bullet}$ for all program points. Furthermore, $\mathrm{MFP}_{\circ},\mathrm{MFP}_{\bullet}$ is precisely $\mathrm{DF}_{\circ},\mathrm{DF}_{\bullet}$.

The first step of our algorithm, Fig. \ref{fig:mfp1}, follows similarly from \cite{popa}.
\begin{figure}[h]
    \centering
    \begin{SpecialText}
    \begin{algorithmic}
\State $W \gets \mathbf{nil}$
\For{\textbf{all} $(\ell,\ell') \in F$}
    \State $W \gets \operatorname{cons}((\ell,\ell'),W)$
\EndFor
\For{\textbf{all} $\ell$ \textbf{in} $F$ \textbf{or} $S$}
    \If{$\ell \in S$}
        \State $\mathrm{Analysis}[\ell] \gets \widehat{\lambda}$
    \Else
        \State $\mathrm{Analysis}[\ell] \gets \bot_{\widehat{L}}$
    \EndIf
\EndFor
\end{algorithmic}
\end{SpecialText}
    \caption{First step of the algorithm}
    \label{fig:mfp1}
\end{figure}

Now that we have initialized $W$ (the worklist) and $\mathrm{Analysis}$, we can begin updating. This step is to iterate through and ``collect'' dataflow values until we reach an equilibrium (fixed point).
\begin{figure}[h]
    \centering
    \begin{SpecialText}
\begin{algorithmic}
\While{$W \neq \mathbf{nil}$}
    \State $(\ell,\ell') \gets \operatorname{head}(W)$
    \State $W \gets \operatorname{tail}(W)$
    \If{$f_{\ell}^{\Gamma(\ell,\ell')}(\mathrm{Analysis}[\ell]) \not\sqsubseteq \mathrm{Analysis}[\ell']$}
        \State $\mathrm{Analysis}[\ell'] \gets \mathrm{Analysis}[\ell'] \sqcup f_{\ell}^{\Gamma(\ell,\ell')}(\mathrm{Analysis}[\ell])$
        \For{\textbf{all} $\ell''$ \textbf{with} $(\ell',\ell'') \in F$}
            \State $W \gets \operatorname{cons}((\ell',\ell''),W)$
        \EndFor
    \EndIf
\EndWhile
\end{algorithmic}
\end{SpecialText}
    \caption{Second step of the algorithm}
    \label{fig:mfp2}
\end{figure}

After the second step, Fig. \ref{fig:mfp2}, we have computed the entry values for each program point ($\mathrm{Analysis}$). All that is left is to apply the transfer functions to each of these to attain the exit values.

Here, our output for exit values differs considerably from \cite{popa}. Since a label can certainly send different values through different flows (i.e. calling multiple functions or having multiple flows with different types), we instead assign to each label $\ell$ a partial function mapping each $\ell'$ such that $(\ell,\ell') \in F$ to the corresponding value it receives from $\ell$ (and we use a placeholder label $end$ representing the end of the program). This is done in Fig. \ref{fig:mfp3}.
\begin{figure}[h]
    \centering
    \begin{SpecialText}
\begin{algorithmic}
\For{\textbf{all} $\ell$ \textbf{in} $F$ \textbf{or} $S$}
    \State $\mathrm{MFP}_{\circ}(\ell) \gets \mathrm{Analysis}[\ell]$
    \State $\mathrm{Outputs}_{\ell} \gets []$
    \If{$\nexists \ell'$ \textbf{with} $(\ell,\ell') \in F$}
        \State $\mathrm{Outputs}_{\ell}[end] \gets f_{\ell}^{\mathcal{N}}(\mathrm{Analysis}[\ell])$
    \Else
        \For{\textbf{all} $\ell'$ \textbf{with} $(\ell,\ell') \in F$}
            \State $\mathrm{Outputs}_{\ell}[\ell'] \gets f_{\ell}^{\Gamma(\ell,\ell')}(\mathrm{Analysis}[\ell])$
        \EndFor
    \EndIf
    \State $\mathrm{MFP}_{\bullet}(\ell) \gets \mathrm{Outputs}_{\ell}$
\EndFor
\end{algorithmic}
\end{SpecialText}
    \caption{Third step of the algorithm}
    \label{fig:mfp3}
\end{figure}

And now, for the most daunting task: the correctness proof.
\begin{theorem}
    The above algorithm terminates, and upon termination satisfies $\mathrm{MFP}_{\circ} = \operatorname{DF}_{\circ}$.
\end{theorem}
\begin{proof} 
Let $P$ be the program being analyzed with an $\mathcal{L}$-labeling.
\\\\\textbf{Termination.} The first and third parts of the algorithm are loops over finite sets, so they trivially terminate.

For the second part, observe that $|W| \leq |\mathcal{L}|^2$. That is, in the worst case there is a flow between every block in the program in both directions. In each iteration with $(\ell,\ell')$, if $f_{\ell}^{\Gamma(\ell,\ell')}(\mathrm{Analysis}[\ell]) \not\sqsubseteq \mathrm{Analysis}[\ell']$ then up to $\mathcal{L}$ (finite number) elements are added to $W$. Furthermore, $\mathrm{Analysis}[\ell']$ is set to $\mathrm{Analysis}[\ell'] \sqcup f_{\ell}^{\Gamma(\ell,\ell')}(\mathrm{Analysis}[\ell])$, which is strictly greater than its previous value since $f_{\ell}^{\Gamma(\ell,\ell')}(\mathrm{Analysis}[\ell]) \not\sqsubseteq \mathrm{Analysis}[\ell']$. Since $\widehat{L}$ satisfies the Ascending Chain Condition, this can only happen a finite number of times.
\\\\\textbf{Correctness.} We first show that the following is an invariant of the loop in the second part of the algorithm.
\begin{equation*}
    \forall \ell \in \mathcal{L},\; \mathrm{Analysis}[\ell] \sqsubseteq \operatorname{DF}_{\circ}(\ell)
\end{equation*}
Observe that, since $\widehat{\lambda} \sqsubseteq \operatorname{DF}_{\circ}(\ell)$ for all $\ell \in S$, the loop invariant holds after the first part of the algorithm, so it holds before the loop begins.

In the loop there are two cases: either a single element is removed from $W$, thus maintaining the invariant, or $\mathrm{Analysis}$ changes for one $\ell'$. Suppose the latter case, with $(\ell,\ell') \in F$ as the flow of the current iteration. Then,
\begin{align*}
    new\mathrm{Analysis}[\ell'] &= old\mathrm{Analysis}[\ell'] \sqcup f_{\ell}^{\Gamma(\ell,\ell')}(old\mathrm{Analysis}[\ell])\\
    &\sqsubseteq \operatorname{DF}_{\circ}(\ell') \sqcup f_{\ell}^{\Gamma(\ell,\ell')}(\operatorname{DF}_{\circ}(\ell))&\qquad& (1)\\
    &= \operatorname{DF}_{\circ}(\ell')&\qquad& (2)
\end{align*}
(1) holds by the loop invariant and monotonicity of $f_{\ell}^{\Gamma(\ell,\ell')}$. (2) holds since $f_{\ell}^{\Gamma(\ell,\ell')}(\operatorname{DF}_{\circ}(\ell)) = \operatorname{DF}_{\bullet}(\ell)(\ell')$ and $\operatorname{DF}_{\circ}(\ell') \sqsupseteq \operatorname{DF}_{\bullet}(\ell)(\ell')$ by definition. Thus, the invariant holds for the loop in the second part.

Now, upon termination of the loop in the second part, we would like to prove
\begin{equation*}
    \forall (\ell,\ell') \in F,\; f_{\ell}^{\Gamma(\ell,\ell')}(\mathrm{Analysis}[\ell]) \sqsubseteq \mathrm{Analysis}[\ell']
\end{equation*}
Assume, for the sake of eventual contradiction, that $\exists (\ell,\ell') \in F$ such that $f_{\ell}^{\Gamma(\ell,\ell')}(\mathrm{Analysis}[\ell]) \not\sqsubseteq \mathrm{Analysis}[\ell']$.

There are two cases to be considered: either $\mathrm{Analysis}[\ell]$ was last updated in the first part, or it was last updated in the second part.

Assume the former. Then, since we have assumed $f_{\ell}^{\Gamma(\ell,\ell')}(\mathrm{Analysis}[\ell]) \not\sqsubseteq \mathrm{Analysis}[\ell']$, the iteration(s) with $(\ell,\ell')$ establish and maintain the invariant $f_{\ell}^{\Gamma(\ell,\ell')}(\mathrm{Analysis}[\ell]) \sqsubseteq \mathrm{Analysis}[\ell']$ (since $\mathrm{Analysis}[\ell]$ was last updated in the first part, so $\mathrm{Analysis}[\ell']$ can grow larger while $\mathrm{Analysis}[\ell]$ remains constant). This is a contradiction.

Thus we are left with the latter case. If $\mathrm{Analysis}[\ell]$ was last updated in the second part, then it must have immediately been followed by adding $(\ell,\ell')$ to $W$. Then, when we have $(\ell,\ell')$ in the iteration, we establish and maintain the invariant $f_{\ell}^{\Gamma(\ell,\ell')}(\mathrm{Analysis}[\ell]) \sqsubseteq \mathrm{Analysis}[\ell']$ just as previously. This is a contradiction.

Thus both cases leave us with a contradiction, proving the statement. Furthermore, we clearly have from this that

\begin{equation*}
    \forall \ell' \in \mathcal{L},\; \mathrm{Analysis}[\ell'] \sqsupseteq \bigsqcup_{(\ell,\ell') \in F} f_{\ell}^{\Gamma(\ell,\ell')}(\mathrm{Analysis[\ell]})
\end{equation*}
Observe that, after the loop in the first part, we maintain the invariant
\begin{equation*}
    \forall \ell \in S,\; \widehat{\lambda} \sqsubseteq \mathrm{Analysis}[\ell]
\end{equation*}
and so putting these together gives
\begin{align*}
    \mathrm{Analysis[\ell]} &\sqsupseteq \widehat{\lambda}_{S}^{\ell} \;\sqcup\; \bigsqcup_{(\ell,\ell') \in F} f_{\ell}^{\Gamma(\ell,\ell')}(\mathrm{Analysis[\ell]})
\end{align*}
Since $\operatorname{DF}_{\circ}$ is the least solution to the above constraint, we then have that
\begin{equation*}
    \operatorname{DF}_{\circ}(\ell) \sqsubseteq \mathrm{Analysis}[\ell]
\end{equation*}
Then, since the (earlier proved) invariant gives us $\mathrm{Analysis}[\ell] \sqsubseteq \operatorname{DF}_{\circ}(\ell)$, we have
\begin{equation*}
    \forall \ell \in \mathcal{L},\; \mathrm{Analysis}[\ell] = \operatorname{DF_{\circ}}(\ell)
\end{equation*}
which directly gives us
\begin{equation*}
    \forall \ell \in \mathcal{L},\; \operatorname{MFP}_{\circ}(\ell) = \operatorname{DF_{\circ}}(\ell)
\end{equation*}
upon termination of the third part.
\end{proof}
After this, we have that $\operatorname{DF}_{\bullet}$ is exactly defined from what we have already computed (which is done in $\operatorname{MFP}_{\bullet}$).

\section{Appendix}

\subsection{The New Property Space: An Alternative Approach}

Consider Section 4.2.1 (and only read this section after having read Section 4.2.1). What if $\dom{f} \not\subseteq \dom{g}$? Then we have points where one (partial) function is defined and the other is not. In the straightforward approach, this would simply render the two partial functions incomparable. Intuitively, however, for $\Delta \rightharpoonup L$, a partial function being undefined at a context $\delta$ simply means that it has no information in that context (i.e. it does not `deal' with that context). Therefore, we may simply think of it as $\bot_L$. Then, we can \emph{induce} a relation for partial functions that achieves effectively the same thing as that for total functions: for $f,g \in S \rightharpoonup M$
\begin{equation*}
    f \preccurlyeq g \;\text{ iff }\; \forall s \in \dom{f},\; f(s) \sqsubseteq_M {\uparrow}g(s)
\end{equation*}
This relaxes the domain inclusion requirement on orderings. Observe:
\begin{proposition}
    The following are equivalent:
    \begin{myitemize}
        \item $\forall s \in \dom{f},\; f(s) \sqsubseteq_M {\uparrow}g(s)$
        \item $\forall s \in S,\; {\uparrow}f(s) \sqsubseteq_M {\uparrow}g(s)$
    \end{myitemize}
\end{proposition}
What is relatively less intuitive is that these are not equivalent to $\forall s \in \dom{g},\; {\uparrow}f(s) \sqsubseteq_M g(s)$ (imagining these functions as graphs helps make a bit more sense of this). 

If this relation seems sketchy to you, then good job! It most certainly is. Attempting to use this would cause a very large contradiction, and that comes from its inability to satisfy antisymmetry.

\begin{lemma}
    $\preccurlyeq$ is not antisymmetric.
\end{lemma}
\begin{proof}
    Let
    \begin{align*}
        f_1,f_2 : S \rightharpoonup M
    \end{align*}
    with $\dom{f_1} \neq \dom{f_2}$ be defined trivially as $f_1(s) = \bot_M$ and $f_2(s) = \bot_M$ over their respective domains. Since their domains are not equal, we have that $f_1 \neq f_2$. Following our definition of $\preccurlyeq$, however, we have that $f_1 \preccurlyeq f_2$ and $f_2 \preccurlyeq f_1$. Thus antisymmetry is violated.
\end{proof}

Fortunately, Proposition 5.1 sets us up for a solution to this problem. We shall instead define $(S \rightharpoonup M) / {\upeq}$ as the set of equivalence classes in $S \rightharpoonup M$ over $\upeq$, where
\begin{equation*}
    f \upeq g \;\text{ iff }\; {\uparrow}f = {\uparrow}g
\end{equation*}
and we are now treating partial functions as their lifted total functions. This is the quotient lattice over $\upeq$.

\begin{proposition}
    $\uparrow$ is a closure operator on $(S \rightharpoonup M,\sqsubseteq)$.
\end{proposition}
\begin{proof}
    Let $f \in S \rightharpoonup M$. Then we clearly have $f \sqsubseteq {\uparrow}f$. Now, let $g \in S \rightharpoonup M$ with $f \sqsubseteq g$. It is clear that $\dom{{\uparrow}f} \subseteq \dom{{\uparrow}g}$. Let $s \in S$. If $s \in \dom{f}$, then $s \in \dom{g}$ so
    \begin{equation*}
        {\uparrow}f(s) = f(s) \sqsubseteq_M g(s) = {\uparrow}g(s)
    \end{equation*}
    If $s \notin \dom{f}$, then
    \begin{equation*}
        {\uparrow}f(s) = \bot_M \sqsubseteq_M {\uparrow}g(s)
    \end{equation*}
    and so both cases together give ${\uparrow}f \sqsubseteq {\uparrow}g$. Finally, it is clear that ${\uparrow}({\uparrow}f) = f$.
\end{proof}

Then, $(S \rightharpoonup M)/{\upeq}$ is isomorphic to the set of all $\uparrow$-closed elements (those $f$ for which ${\uparrow}f = f$). That is to say,
\begin{equation*}
    (S \rightharpoonup M)/{\upeq} \;\cong \{f \in S \rightharpoonup M \mid {\uparrow}f = f\}
\end{equation*}
over $\sqsubseteq$. We will refer to $((S \rightharpoonup M)/{\upeq},\preccurlyeq)$ by $(S \rightharpoonup M)^{{\upeq}}$ for the remainder of this paper. One may question why we study $(S \rightharpoonup M)^{\upeq}$ as opposed to the set of closed partial functions (i.e. the set of total functions). This question is answered later when comparing the two approaches. 

Now, the previously defined relation $\preccurlyeq$ is indeed a partial ordering over $(S \rightharpoonup M)^{\upeq}$, and in fact forms a complete lattice.

Now, we see that the join operation over $(S \rightharpoonup M)^{\upeq}$ becomes
\begin{equation*}
    (f \curlyvee g)(s) = \begin{cases}
        f(s) \sqcup_M g(s) & \text{if } s \in \dom{f} \wedge s \in \dom{g}\\
        f(s) & \text{if } s\in\dom{f}\wedge s\notin\dom{g}\\
        g(s) & \text{if } s\notin\dom{f}\wedge s\in\dom{g}\\
        \text{undefined} & \text{otherwise}
    \end{cases}
\end{equation*}
This operation, in a similar way, is derived from the join operation over total functions. The join operation over total functions is defined as $(f \sqcup g)(s) = f(s) \sqcup_M g(s)$. The following lemma shows that the join operation over partial functions may be treated `similarly' to that of total functions.
\begin{lemma}
    ${\uparrow}(f\curlyvee g) \upeq {\uparrow}f \sqcup {\uparrow}g$
\end{lemma}
\begin{proof}
    This proof is left for the reader.
\end{proof}
Lemma 5.2 also shows that the join operation here is (modulo $\upeq$) the same as the previous approach.
\begin{corollary}
    $f\curlyvee g \upeq f \sqcup g$
\end{corollary}
\begin{proof}
    ${\uparrow}(f\curlyvee g) \upeq f\curlyvee g$. Lemma 5.2 then gives $f\curlyvee g \upeq {\uparrow}f \sqcup {\uparrow}g$. By the definition of $\sqcup$ in the straightforward approach, we then have $f\curlyvee g \upeq f \sqcup g$.
\end{proof}
\begin{lemma}
    $f\curlyvee g$ is the least upper bound of $f,g \in (S \rightharpoonup M)^{\upeq}$.
\end{lemma}
The proof is similar to that of Lemma 4.1, except using the new partial ordering. Similarly, $(S \rightharpoonup M)^{\upeq}$ is a complete lattice, which follows similarly from Corollary 4.1. The last necessary condition is that our property space $(\Delta \rightharpoonup L)^{\upeq}$ satisfies the Ascending Chain Condition.

\begin{lemma}
    $\widehat{L}$ satisfies the Ascending Chain Condition.
\end{lemma}
\begin{proof}
    Suppose $(\widehat{l}_i)_{i \in \mathcal{I}}$ is a chain in $\widehat{L}$, and pick an arbitrary $i,i' \in \mathcal{I}$ with $\widehat{l}_i \sqsubseteq \widehat{l}_{i'}$. This implies 
    \begin{equation*}
        \forall \delta \in \dom{\widehat{l}_i},\; \widehat{l}_i(\delta) \sqsubseteq_L {\uparrow}\widehat{l}_{i'}(\delta)
    \end{equation*}
    and by Proposition 4.1 further implies
    \begin{equation*}
        \forall \delta \in \Delta,\; {\uparrow}\widehat{l}_i(\delta) \sqsubseteq_L {\uparrow}\widehat{l}_{i'}(\delta)
    \end{equation*}
    Assume $\widehat{l}_i \neq \widehat{l}_{i'}$. Then, there is some $\delta \in \Delta$ such that ${\uparrow}\widehat{l}_{i'}(\delta)$ is strictly greater than ${\uparrow}\widehat{l}_i(\delta)$. Since $L$ satisfies the Ascending Chain Condition, this can only happen a finite number of times for that $\delta$, and since $\Delta$ is finite (given our earlier assumption), there are only finitely many $\delta \in \Delta$ to be ``increased.'' 
    
    These two together clearly give that $\widehat{L}$ satisfies the Ascending Chain Condition.
\end{proof}
In fact, Lemma 4.2 holds in this approach as well following a similar proof using the alternative partial order.

\subsubsection{Comparing the Approaches}

While there seem to be many similarities, the two approaches certainly have their differences. Throughout this section, we will denote the straightforward approach by $(S \rightharpoonup M, \sqsubseteq)$ and the alternative approach by $((S \rightharpoonup M)^{\upeq}, \preccurlyeq)$.

One may begin to think that these are equivalent because their join operation is the same. This is not true, however, as we shall now see.
\begin{lemma}
    $((S \rightharpoonup M)^{\upeq}, \preccurlyeq)$ is order-isomorphic to $S \rightarrow M$.
\end{lemma}
\begin{proof}
    Let $[f] \in (S \rightharpoonup M)^{\upeq}$ be the \emph{equivalence class} represented by $f$ (that is, $\{g \mid f \upeq g\}$). Then, define $\varphi : (S \rightharpoonup M)^{\upeq} \rightarrow (S \rightarrow M)$ by
    \begin{equation*}
        \varphi([f]) = {\uparrow}f
    \end{equation*}
    and observe that $\varphi$ is well-defined since $\upeq$ is an equivalence relation. Furthermore, $\varphi$ is clearly injective since if $[f],[g]$ map to the same element, then ${\uparrow}f = {\uparrow}g$ which means $[f] = [g]$).

    Now, let $f \in S \rightarrow M$. Then, since $(S \rightarrow M) \subseteq (S \rightharpoonup M)$, we have that $[f] \in (S \rightharpoonup M)^{\upeq}$, and
    \begin{equation*}
        \varphi([f]) = f
    \end{equation*}
    so $\varphi$ is surjective, and thus a bijection.

    Now we show that $\varphi$ is monotone. Let $f,g \in (S \rightharpoonup M)^{\upeq}$ with $f \preccurlyeq g$. Then, by Proposition 5.1,
    \begin{equation*}
        \forall s \in S,\; {\uparrow}f(s) \sqsubseteq_M {\uparrow}g(s)
    \end{equation*}
    so ${\uparrow}f \sqsubseteq_{S \rightarrow M} {\uparrow}g$. Therefore,
    \begin{equation*}
        \varphi(f) \sqsubseteq_{S \rightarrow M} \varphi(g)
    \end{equation*}
    and we are done.
\end{proof}
This result begs the question: why even study such a relation in such a complex way if we can study it in a much simpler way through total functions? The main merit for this is to illustrate the \emph{exact} implementation/computations done in the program (which deals with partial functions), and that they are correct ``by construction'' in some sense (by following the proofs in the implementation). It is indeed correct, however, that it is much easier to simply study the total functions (which has been done quite extensively). For these reasons, we call $(S \rightharpoonup M)^{\upeq}$ the \emph{derived-total functions}.

Furthermore, $((S \rightharpoonup M)^{\upeq}, \preccurlyeq)$ has a much more flexible and rich structure than $(S \rightharpoonup M, \sqsubseteq)$, which has quite a rigid structure. This seems unintuitive since the former is (isomorphic to) a sublattice of the latter, but ``flattening out'' these ``buffers'' allows much more inter-chain comparisons. That is, chains are less isolated in the former lattice.

From now on, we will call $T = (S \rightharpoonup M)^{\upeq}$ (for \emph{(derived)-total}) and $P = S \rightharpoonup M$ (for \emph{partial}).

\begin{proposition}
    Let $f,g \in T$. $f \preccurlyeq g$ and $\dom{f} \subseteq \dom{g}$ if and only if $f \sqsubseteq g$.
\end{proposition}
We omit the proof as it is direct from definitions.

\begin{proposition}
    For $f \in P$, ${\uparrow}f = \bigsqcup_{g \in [f]} g$ where $[f]$ is the $\upeq$-equivalence class of $f$ in $T$.
\end{proposition}
\begin{proof}
    Let $g \in [f]$ be the function where $g = {\uparrow}f$ in $P$. That is, the $g$ which is defined on all of $S$, representing $f$. Then, we clearly have that $g \sqsupseteq g'$ for all $g' \in [f]$, so $\bigsqcup_{g' \in [f]} g' = g = {\uparrow}f$.
\end{proof}

\begin{corollary}
    For $[f] \in T$, $\operatorname{max}_{\sqsubseteq} [f]$ exists, and ${\uparrow}f = \operatorname{max}_{\sqsubseteq} [f]$.
\end{corollary}
\begin{proof}
    Immediate from the proof of Proposition 5.4.
\end{proof}

Something that is harder to see when studying through the set of total functions, however, is the fluidity between these two sets. In fact, they induce multiple very natural Galois Connections.

Define $\alpha : P \rightarrow T$ by $\alpha(f) = {\uparrow}f$ and $\gamma : T \rightarrow P$ as the natural inclusion $\gamma(f) = \operatorname{max}_{\sqsubseteq}f$.
\begin{lemma}
    $(P, \sqsubseteq) \galoiS{\alpha}{\gamma} (T, \preccurlyeq)$ is a Galois Connection (with $\alpha \circ \gamma = 1$).
\end{lemma}
\begin{proof}
    Let $f \in P$ and $g \in T$. First assume $\alpha(f) \preccurlyeq g$. Then, ${\uparrow}f \preccurlyeq g = {\uparrow}g$, so by Proposition 5.1, $f \preccurlyeq {\uparrow}g = g$ (and we know $\dom{f} \subseteq \dom{{\uparrow}g}$), and thus
    \begin{equation*}
        f \sqsubseteq g = \gamma(g)
    \end{equation*}
    by Proposition 5.3. Since all of these steps are equivalences, the backwards direction holds as well.

    We now show that $\alpha \circ \gamma = 1$ (which is equivalent to $\alpha$ is surjective and also equivalent to $\gamma$ is injective). Let $f \in T$. Then,
    \begin{equation*}
        (\alpha \circ \gamma)(f) = \alpha(\operatorname{max}_{\sqsubseteq}f) = {\uparrow}\operatorname{max}_{\sqsubseteq}f = \operatorname{max}_{\sqsubseteq}f = {\uparrow}f = f
    \end{equation*}
    in $T$. The second to last equality uses Corollary 5.2.
\end{proof}
So really what is happening between these two approaches is that the alternative approach is simply an \emph{abstraction} of the straightforward approach, which is not entirely easy to see from the difference in domain inclusion requirements.
\begin{proposition}
    $\alpha,\gamma$ are lattice homomorphisms.
\end{proposition}
\begin{proof}
    Let $f,g \in P$. Then,
    \begin{equation*}
        \alpha(f\sqcup g) = {\uparrow}(f\sqcup g) = f\sqcup g = {\uparrow}f \sqcup {\uparrow}g = \alpha(f) \sqcup \alpha(g) = \alpha(f) \curlyvee \alpha(g)
    \end{equation*}
    in $T$ since $\uparrow$ can be added or removed without affecting anything. The last equality follows from Corollary 5.1.

    Now, observe
    \begin{equation*}
        \gamma(f\curlyvee g) = \operatorname{max}_{\sqsubseteq}(f\curlyvee g) = {\uparrow}(f\curlyvee g) = {\uparrow}(f\sqcup g)
    \end{equation*}
    by Corollary 5.2 and 5.1. Now we are in $P$, so we can not add and remove $\uparrow$ as we please. Proposition 4.1, however, gives us
    \begin{equation*}
        {\uparrow}(f\sqcup g) = {\uparrow}f \sqcup {\uparrow}g = \operatorname{max}_{\sqsubseteq}(f) \sqcup \operatorname{max}_{\sqsubseteq}(g) = \gamma(f) \sqcup \gamma(g)
    \end{equation*}
    The meet-homomorphism requirement for $\alpha$ and $\gamma$ is left for the reader.
\end{proof}

\section{Implementation}

We may now proceed to implementation, which serves to be the most flexible part. So long as the user has some way of specifying the necessities of an Implicit Monotone Framework, we have what we need.

A specification language for Monotone Frameworks, FlowSpec, has been proposed and developed \cite{flowspec}. The major benefit to this approach is the succinct specification of flows and transfer functions. They also provide an MFP solver for Monotone Frameworks. This could suffice for solving interprocedural analyses, however it makes the classification of flows (normal, calling, returning) a bit more difficult. FlowSpec provides a more succinct specification than Flix \cite{flix}.

Our implementation utilizes ANTLR4 \cite{antlr} to achieve a Java implementation, and we take a generative approach like what is presented in \cite{pag}. There does exist a java library for program analysis \cite{wala}, however its flexibility for specifying grammars is limiting as opposed to specifying grammars with ANTLR4, and it is much more work on the user.

\section{Conclusion}

In this paper we went over the theoretical framework for language-general dataflow analysis through user specification. We varied from the classical Monotone Framework \cite{Kam1977} and Embellished Monotone Framework \cite{popa} through the introduction of \emph{tagged} flows and inference from \emph{Implicit} Monotone Frameworks.

We also presented a new approach to contexts through the use of partial functions. We developed these new property spaces rigorously to prove correctness of the implementation ``by construction.'' Furthermore, we compared two partial function property spaces, and showed that they work quite well together.

The Galois Connection formed between the derived-total and partial functions hints at some possibility of a meta-analysis. Such an analysis could be inspired by the ``$\text{A}^2\text{I}$'' framework for meta-analyses presented in \cite{cousota2i}.

\section{Acknowledgements}

This project was completed as a senior capstone project for Kansas State University. I would like to give a big acknowledgement to Dr. Torben Amtoft for his invaluable feedback and discussions throughout this project. If it was not for Dr. Amtoft, I would not have indulged into the beautiful field of program analysis. I would also like to thank Dr. Pietro Poggi-Corradini for his feedback and encouragement on this paper.

\bibliographystyle{plain} 
\bibliography{refs} 

\end{document}